\documentclass[12pt,a4paper]{amsart} 
\usepackage[utf8]{inputenc}
\usepackage[a4paper,inner=2.8cm,outer=2.5cm,top=2.8cm,bottom=2.5cm]{geometry}                
\usepackage{amsmath,amssymb,amscd,amsthm,amsfonts,anyfontsize,color,enumerate,fix-cm,layout,lipsum,lpic,mathrsfs,mdwlist,stmaryrd,tensor,tikz,thmtools,xspace,comment}
\usepackage{pdfpages}
\usepackage{graphicx}
\usepackage{hyperref}
\usepackage[all]{xy}
\usepackage{mathtools}
\usepackage[bottom]{footmisc}
\usepackage{ifthen}
\usepackage{xfrac}
\theoremstyle{plain}                 
\newtheorem{theorem}{Theorem}[section]

\newtheorem{lemma}[theorem]{Lemma}        
\theoremstyle{definition}           
\newtheorem{definition}[theorem]{Definition}    
 
\theoremstyle{remark}       
\newtheorem{remark}[theorem]{Remark}    
\hypersetup{colorlinks=true,linkcolor=blue,citecolor=purple,filecolor=magenta,urlcolor=cyan}

\newcommand{\R}{\mathbb{R}}

\def\bea{\begin{eqnarray}}
\def\eea{\end{eqnarray}}
\def\nn{\nonumber}

\begin{document}

\title{Vertex electrical model: \\lagrangian and non-negative properties}

\author[D.~Talalaev]{D.~Talalaev}
\address{D.~T.: Lomonosov Moscow State University, Moscow, Russia; and Centre of Integrable Systems, P.~G.~Demidov Yaroslavl State University, Sovetskaya 14, 150003, Yaroslavl, Russia}
\email{dtalalaev@yandex.ru}

\begin{abstract}
This note is a modest addition to the work \cite{BGKT}. Here we construct an embedding of the space of electrical networks to the totally non-negative Lagrangian  Grassmannian in a generic situation with the help of the technique of vertex integrable models of statistical mechanics.

\end{abstract}

\maketitle

\tableofcontents

\hspace{0.1in}

\section{Introduction}
The question of embedding the space of electrical networks to a Lagrangian non-negative Grassmannian already has a rather rich history and should be perceived in the context of other archetypal models on weighted graphs. The latter include the dimer model and the Ising model. For the first of them, there is an embedding of the model space into a non-negative Grassmannian \cite{P}, for the second - into a non-negative orthogonal Grassmannian \cite{GP}. The electrical network model is related to the Lagrangian version of the non-negative Grassmannian. This statement has been proven with the help of various techniques in the works \cite{LP}, \cite{L}, \cite{CGS} and \cite{BGKT}.
Here we apply the technique of integrable vertex models of statistical physics to the Grassmannian embedding problem. The actual idea of the vertex representation is to express the observables in a factorized form, each of whose multipliers satisfies local controllable properties. These are usually versions of the Yang-Baxter equation, the local Yang-Baxter equation. The vertex representation of electrical networks involves the solution of the local Yang-Baxter equation encoding Ohm's law. This solution is related to the so-called electrical solution of the Zamolodchikov tetrahedron equation \cite{Zam}. This representation was found in Sergeyev's work \cite{Ser} and led to the modern understanding of electrical manifolds proposed in the work \cite{GT}.

By an electrical network $e$, we mean a connected graph $\Gamma$ without loops with a function of weights on the edges $c:E\to \mathbb{R}^+$ (the conductivity function) and a  subset of boundary vertices   
$V_B=\{1,\ldots,n\} \subset V=\{1,\ldots,N\}.$ Let $c_{ij}$ be the conductivity of an edge connecting vertices $i$-th and $j$-th which is zero if the vertices are not connected.
\begin{definition}
The Kirchhoff matrix of a network is a matrix
\bea
T_{ij}=\left\{ 
\begin{array}{ccc}
-c_{ij} &\mbox{if}& i\ne j\\
\sum_{k\ne i} c_{ik} &\mbox{if}& i=j
\end{array}
\right.\nn
\eea
\end{definition}
The network response matrix $M_R(e)$ relates the currents and potentials at the boundary points of the network   
\bea
I=M_R U,\nn
\eea
where  $I:V_{B}\to  \mathbb{R}$ are the currents and $U:V_{B}\to  \mathbb{R}$ are the potentials in boundary points.    
Let us present the Kirchhoff matrix in a block form where the $n\times n$ block $A$ corresponds to the subset of  boundary points:
\bea
T=\left(
\begin{array}{cc}
A & B\\
B^T & C
\end{array}
\right)\nn
\eea
then the response matrix is given by an expression
\bea
M_R=A-B C^{-1} B^T,\nn
\eea
which is just the Schur complement of $C.$ We denote by $E_n$ the variety of electrical networks with $n$ boundary points. The main statements of this note are the theorem \ref{lagr} where we construct an embedding of the subvariety of electrical networks on standart graphs $\subset E_n$ to the lagrangian grassmanian $LG(n-1)$ and the theorem \ref{nonneg} which demonstrate the total non-negativity of this  embedding in the case of odd $n.$

\section{Vertex model}
\label{vertex}
The set of standard graphs is a subset of the set of critical graphs (see \cite{CIM}, \cite{CIM1}). They are defined inductively. The first few examples of standard graphs are presented on Figure \ref{stgr}. 
\begin{figure}[h!]
\centering
\includegraphics[width=80mm]{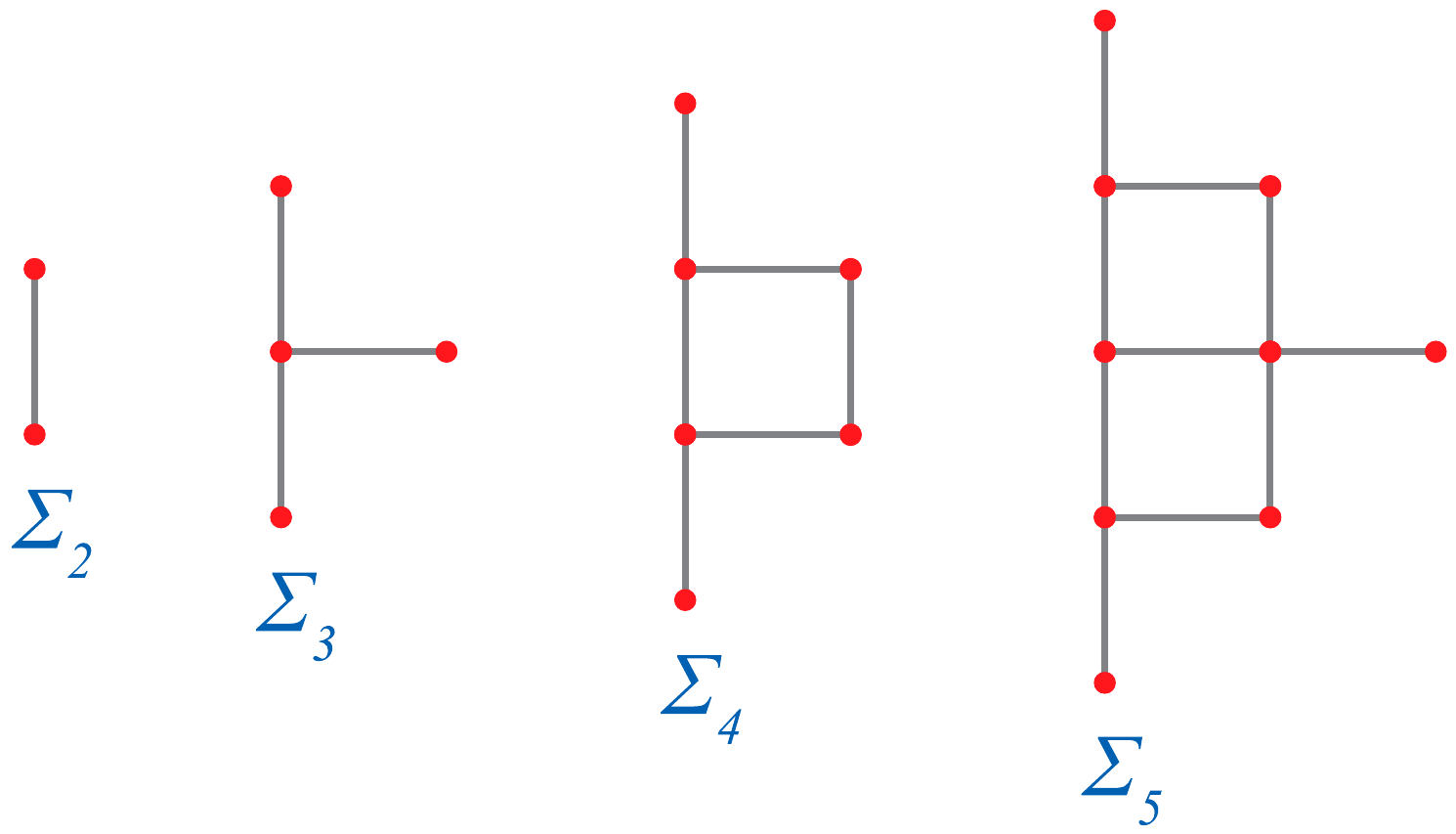}
\caption{Standard graphs}
\label{stgr}
\end{figure} 
In this paper we deal only with standard graphs. In the problem of parametrization of electrical networks, they represent, in a certain sense, a point of generic position.

For $e\in E_n$ on a standard graph the authors of \cite{GT} defined  the boundary partition function $M_B(e)$. It is a matrix which depends on at most $n(n-1)/2$ parameters, these parameters are the conductivities of the edges of the network.  According to \cite{GT} such a matrix  belongs to the symplectic group $Sp(n)$. 

Let us introduce the notations
\bea
T_{2n}=\left(
\begin{array}{ccccccc}
1 & 0 & 0 &\cdots & 0 & 0 &\cdots\\
0 & 0 & 0 & \cdots & 1 & 0 & \cdots \\
0 & 1 & 0 & \cdots & 0 & 0 & \cdots \\
0 & 0 & 0 & \cdots & 0 & 1 & \cdots \\
0 & 0 & 1 & \cdots & 0 & 0 & \cdots \\
\vdots &  \vdots & \vdots & \ddots & \vdots & \vdots & \vdots\\
0 & 0 & 0& \cdots & 0 & 0 & \ddots
\end{array}
\right);\qquad
S_n=\left(
\begin{array}{ccccc}
1 & 0 &0 &  \cdots & -1\\
-1 & 1 & 0 &  \cdots & 0 \\
0 &-1 & 1 &  \cdots & 0 \\
\vdots &  \vdots & \vdots & \ddots & \vdots\\
0 & 0 & 0 & \cdots & 1
\end{array}
\right).\nn
\eea
One of the main results from \cite{GT} is the following
\begin{lemma} \label{lemma:w1w2}
For an electrical network $e \in E_n$ on the standard critical graph $\Sigma_n$ the row spaces of the matrices
\bea
W_1=(S_n,M_R)\quad \mbox{and} \quad
W_2=(M_B,Id_n)S_{2n} T_{2n}\nn
\eea
define the same point in $\mathrm{Gr}(n-1,2n)$.
\end{lemma}

\subsection{Lagrangian property} 
The modified boundary measurement matrix $\check{M}_B(e)$ for the standard graph is given by
\bea
\check{M}_B=\prod_{i<j}\check{\phi}_{j-i,j-i+1}\left((-1)^{i+j}r_{ij}^{(-1)^{i+j}}\right).\nn
\eea
Here  $E_{ij}$ are just elementary matrices with $1$ on the $i,j$-th place and the multipiers in the product are given by
\bea
\check{\phi}_{i i+1}(s)=1-s(E_{ii}+E_{i,i+1}-E_{i+1,i}-E_{i+1,i+1}) \in Mat_{n},\nn
\eea
the parameters $r_{ij}=1/c_{ij}$ are the resistances of the corresponding edges.
Then 
\bea
M_B=\omega_0 \check{M}_B,\nn
\eea
where
$\omega_0$ represents the longest element in the Weil group
\bea
\omega_0=\left(
\begin{array}{ccccc}
0 & 0 & \cdots & 0& 1\\
0 & 0 &  \cdots & 1 & 0\\
\vdots & \vdots & \ddots & \vdots & \vdots\\
0 & 1 & \cdots & 0& 0\\
1 & 0 &  \cdots & 0 & 0
\end{array}
\right),\nn
\eea

\begin{lemma}
The row space of the matrix 
\bea
W_0=(M_B,Id_n)\nn
\eea
is an isotropic subspace with respect to the bilinear form
\bea
\eta=\left(
\begin{array}{cc}
g & 0\\
0 & g
\end{array}
\right);
\eea
where
\bea
g=\left(
\begin{array}{ccccc}
0 & -1 & 1 & \cdots &(-1)^{n-1}\\
1 & 0 & -1 & \cdots & (-1)^n\\
-1 & 1 & 0 & \cdots & (-1)^{n-1}\\
\vdots & \vdots & \vdots & \ddots & \vdots\\
(-1)^{n} & (-1)^{n-1} & \cdots & \cdots & 0
\end{array}
\right).
\eea
\end{lemma}
\begin{proof}

We want to show that
\bea
W_0\eta W_0^t=0\nn
\eea
which is equivalent to the following
\bea
M_Bg M_B^t+g=0.\nn
\eea

By Theorem 4.11 of \cite{GT} the modified boundary measurement matrix $\check{M}_B(e)$
preserves the skew-symmetric form $g$
\bea
\label{sp}
\check{M}_B g \check{M}_B^t = g.
\eea
This equality follows from the same one for each generator 
$\check{\phi}_{i,i+1}(s).$ 

Therefore
\bea
M_B g M_B^t= \omega_0 \check{M}_B g \check{M}_B^t  \omega_0=\omega_0 g \omega_0=-g.\nn
\eea

\end{proof}

The row space of  $W_0$ contains the vector $\xi=(1,\ldots, 1)\in  \R^{2n}$ which generates the left kernel of $S_{2n}.$

\begin{lemma}
The row space of the matrix $W_2$  lies in a subspace $\widetilde{V}\in \R^{2n}$ of dimension $2n-2,$ i.e. $W_2\in\mathrm{Gr}(n-1,\widetilde{V}).$
\end{lemma}
\begin{proof}
Let $\mu = (1,-1,1,\ldots,(-1)^{n+1})\in\R^{n}$, then
\bea
\check{\phi}_{i,i+1} \mu=\mu.\nn
\eea
Hence
\bea
\check{M}_B \mu=\mu;\qquad M_B \mu=\omega_0 \mu=\pm \mu;\nn
\eea
with $+$ for odd $n$ and $-$ for even $n$.
Hence the row space of $W_0$ lies in the orthogonal complement to the $2n$-vector 
$$w=\left(\begin{array}{c} \mu \\ \mp \mu\end{array}\right) = (1,-1,1,-1,\ldots,1,-1).$$
The kernel of $S_{2n}$ lies in this complement. Hence the image of $w^{\perp}$ under the action of $S_{2n}T_{2n}$ is a $(2n-2)$-dimensional space which we denoted by $\widetilde{V}$.
\end{proof}
Note that $\xi^t \eta=w^t$ for  even $n$ and $\xi^t \eta=0$ for $n$ odd. Therefore $\xi^t \eta v=0$ $\forall v\in w^\perp$ and
\bea
\xi\in Ann (\eta|_{w^{\perp}}),\nn
\eea

therefore the form $\eta$ can be projected to the quotient $w^{\perp}/\langle\xi\rangle$.

From now on we will work inside $w^{\perp}/\langle\xi\rangle$. It is convenient to prove Lagrangian and nonnegativity properties in the appropriate bases of this space. Summarizing the above arguments we obtain 
\begin{theorem}
\label{lagr}
The projection of the row space of $W_0$ to the space $w^{\perp}/\langle\xi\rangle$ is a lagrangian subspace with respect to the reduction of $\eta.$
\end{theorem}

\subsection{Nonnegativity}
In this subsection we present a demonstration of the total nonnegative property of the subspace corresponding to an electrical network.  Strictly speaking, this method works only for standard graphs moreover only for the case of odd $n$. But the arguments seem to allow generalization.

 Choose a basis  in $w^{\perp}$
\bea
f_i=e_i+e_{i+1},\qquad i=1,\ldots, 2n-1.\nn
\eea

The matrix 
\bea
(\check{M}_B,Id_n)=\omega_0(M_B,Id_n)\left(\begin{array}{c} Id_n\\ \omega_0\end{array}\right)
\label{mcheck}
\eea
represents the same point in $\mathrm{Gr}(n,2n)$ as $(M_B,Id_n)$ after an appropriate change of basis in $\R^{2n}.$

\begin{remark}
The representation \eqref{mcheck} is important for us, it provides an explicit expression for the point of $\mathrm{Gr}(n,2n)$ in terms of the vertex representation of boundary measurements matrix. \eqref{mcheck} is suitable to relate the nonegativity property of the whole matrix with its nontrivial block. In what follows we adapt this strategy making a reduction to $w^{\perp}/\langle\xi\rangle.$
\end{remark}

Note that $\check{M}_B$ preserves $\mu$ on the right and the vector $\zeta=(1,1,1,\ldots)\in \R^{n}$ on the left:
\bea
\zeta^t \check{M}_B=\zeta^t;\qquad \check{M}_B \mu=\mu.\nn
\eea
Hence its left action preserves the subspace $V_\mu=\mu^{\perp}.$ Let us choose the following basis in $V_\mu$
\bea
f_i=e_i+e_{i+1},\qquad i=1,\ldots, n-1.\nn
\eea
The operators $\chi_{i i+1}(t)=\check{\phi}_{i i+1}(t)|_{V_\mu}$ in this basis (let us make attention that we act on vectors on the left) takes the form
\bea
\chi_{i i+1}(t)=1+t(E_{i+1 i}-E_{i-1 i})\nn
\eea
with the exceptions:
\bea
\chi_{12}(t)=1+t E_{21}\nn
\eea
and
\bea
\chi_{n-1,n}(t)=1-t E_{n-2,n-1}.\nn
\eea
Also let us define the matrix
\bea
\Delta=Diag(1,-1,-1,1,1,-1,\ldots)\in Mat_{n-1}.\nn
\eea
Let us introduce the next notation:
\bea
u_{i i+1}(t)=1+t(E_{i i-1}+E_{i i+1})\nn
\eea
with the appropriate exceptions for $i=1$ and $n-1.$
This matrix is totally nonnegative for $t\ge 0$ for obvious reasons.
\begin{lemma}
\label{lem:positive}
The matrix 
\bea
\Delta\chi_{i,i+1}(t)\Delta\nn
\eea
takes the form
\bea
\Delta \chi_{i i+1}(t) \Delta=u_{i i+1}((-1)^i t).\nn
\eea

\end{lemma}
\begin{remark}
The similar formulas for the generators of the symplectic group was obtained in \cite{BGG} in their exploring of the positivity properties of the electrical group.
\end{remark}

\begin{lemma}
\label{vertextheor}
The matrix 
\bea
\mathcal{M}_B=\Delta \check{M}_B|_{V_\mu} \Delta\nn
\eea
in the basis $\{f_i\}$ is totally nonnegative.
\end{lemma}
\begin{proof}
The proof is a consequence of the Binet-Cauchy formula and of the observation that this matrix is a product of totally nonegative factors from Lemma \ref{lem:positive}. The subspace $V_\mu$ is invariant with respect to the action of all factors of the decomposition
\bea
\mathcal{M}_B=\Delta \check{M}_B|_{V_\mu} \Delta&=&\Delta \left(\prod_{i<j}\check{\phi}_{j-i,j-i+1}\left((-1)^{i+j}r_{ij}^{(-1)^{i+j}}\right)\right)|_{V_\mu}\Delta\\
&=&\Delta\prod_{i<j}\chi_{j-i,j-i+1}\left((-1)^{i+j}r_{ij}^{(-1)^{i+j}}\right)\Delta\\
&=&\prod_{i<j}u_{j-i,j-i+1}\left(r_{ij}^{(-1)^{i+j}}\right).
\nn
\eea
Lemma \ref{lem:positive} provides that all factors are totally nonnegative.
\end{proof}

Here we need a version of the lemma 3.9 from \cite{P}.
\begin{lemma}
\label{post}
Let $A$ be a totally nonnegative matrix $n\times n$ than the matrix $n\times 2n$ 
\bea
\psi(A)=(Id_n, \omega_0 D_n A)\nn
\eea
gives a point in the nonnegative Grassmannian.
\end{lemma}
An $n=4$ example of $\psi(A)$ is given by
\bea
\psi(A)=\left(
\begin{array}{cccccccc}
1 & 0 & 0 & 0 & -a_{41} & -a_{42} & -a_{43} & -a_{44} \\
0 & 1 & 0 & 0 & a_{31} & a_{32} & a_{33} & a_{34} \\
0 & 0 &1 & 0 & -a_{21} & -a_{22} & -a_{23} & -a_{24} \\
0 & 0 & 0 & 1 & a_{11} & a_{12} & a_{13} & a_{14} 
\end{array}
\right).\nn
\eea
\begin{proof}
In fact we can demonstrate a stronger statement which identifies some minors of $A$ with principal minors of $\psi(A).$ Let $I=\{i_1,\ldots,i_k\}$ be a subset in $[n].$ We denote by $Inv(I)\subset [n]$ the subset 
\bea
Inv(I)=\{n-i_k+1,\ldots,n-i_1+1\},\nn
\eea
and by $\widetilde{I}$ the shifted subset in $[2n]$
\bea
\widetilde{I}=\{i_1+n,\ldots,i_k+n\}.\nn
\eea
Then the following is true
\bea
\label{minors}
\Delta_{[n]\backslash I\cup\widetilde{J}}\psi(A)=\Delta_{Inv(I)}^J A.
\eea
Let us demonstrate this by two steps. The first step considers the set $I$ concentrated at the end of $[n]$, that is of the type $I=\{n-k+1,\ldots,n-1,n\}.$ In this case the principal minor $\Delta_{[n]\backslash I\cup\widetilde{J}}\psi(A)$ coincides with the minor $\Delta_{I}^{\widetilde{J}}\psi(A)$. The last one coincide up to the sign with the minor $\Delta_{I}^J A.$ Let us show that the sign is the same. The lower $k$ rows of $\omega_0 D_n A$ are just the first $k$ rows of $A$ in altered order with some rows with altered signs. The length of the longest element in $S_k$ is $\frac{k(k-1)}{2}$. The number of rows with altered signs is $\left[\frac k 2\right].$ They coinside and hence the signs of minors coincide. 

Now let us demonstrate the statement for minors for generic $I$ by induction with respect to the lexicographic order. Let $I=\{i_1,\ldots,i_k\}$ be such that for all $I'>I$ the formula \ref{minors} is true. For examle the statement is true for the $I'=\{i_1,\ldots,i_s+1,\ldots,i_k\}$ (with the shift in only one index).
We have 
\bea
\Delta_{[n]\backslash I'\cup\widetilde{J}}\psi(A)=\Delta_{Inv(I')}^J A.\nn
\eea
We know that $\Delta_{[n]\backslash I\cup\widetilde{J}}\psi(A)$ and $\Delta_{Inv(I)}^J A$ coincide up to a sign. Let us remark the passing  from $I'$ to $I$ the sign of the left side expression changes twice, the first one due to the fact that the neighboring rows of right block of $\psi(A)$ have different signs, and the secondly: the change of the neighboring columns in the left block of $\psi(A)$ also changes the sign.
\end{proof}

\begin{theorem}
\label{nonneg}
Let us consider the subspace $V_0\in \R^{2n}$ generated by 
\bea
 \{\pm f_{2n-1},\ldots, f_{n+5}, -f_{n+4},-f_{n+3},f_{n+2},f_{n+1},
 f_1,-f_2,-f_3,f_4,f_5,-f_6,\ldots,\pm f_{n-1}\}\nn
\eea
 For $n$ odd the projection of  the space $w^{\perp}/\langle\xi\rangle$ to $V_0$ is an isomorphism.   The image of the row space of $(\check{M}_B,Id_n)$ after this projection is represented by the matrix
\bea
\psi(\Delta \check{M}_B|_{V_\mu}\Delta).\nn
\eea
which corresponds to a nonnegative point of the  Grassmannian. 
\end{theorem}
\begin{proof}
The vector $\xi\in w^{\perp}$ has a representation
$\xi=\sum_{i=1}^n f_{2i-1}.$ The projection of $w^{\perp}$ to the subspace generated by $\{f_i\}$ except $f_n$ induces an isomorphism on $w^{\perp}/\langle\xi\rangle.$ This projection could be described by the action of the blocks of the matrix $(\check{M}_B,Id_n)$ on the basis on the left. This action after the sign change given by the matrix $\Delta$ is represented by the matrix
\bea
(\Delta\check{M}_B|_{V_\mu}\Delta,Id_{n-1})=(\mathcal{M}_B,Id_{n-1}).\nn
\eea
The following matrix multiplication
\bea
\omega_0 D_{n-1} \times(\mathcal{M}_B,Id_{n-1}) \times  \left(\begin{array}{cc} 0 & Id_{n-1} \\ \omega_0 D_{n-1} & 0\end{array} \right)=
(Id_{n-1},\omega_0 D_{n-1} \mathcal{M}_B)\nn
\eea
transforms the matrix $(\mathcal{M}_B,Id_{n-1})$ to the form $\psi(\mathcal{M}_B).$ Let us comment that the left multiplication does not change the point of the Grassmannian and the right multiplication is just the change of the basis in $V_0,$ it transforms the matrix of the symplectic form and the subspace remains Lagrangian.
\end{proof}

\section*{Acknowledgements} 
I am grateful to B. Bychkov, V. Gorbunov and A. Kazakov for numerous discussions and especially to B. Bychkov for a careful reading of the text.

The work was partially supported by the Basis foundation, the grant Leader (Math) 20-7-1-21-1.

\end{document}